%% file: subregularUniversality.tex
\documentclass[runningheads]{llncs}
\usepackage[T1]{fontenc}
\usepackage{graphicx}
\usepackage{amsmath,amssymb,enumitem}
\usepackage[utf8]{inputenc}
\usepackage{todonotes}

\usepackage{thmtools}
\usepackage{chngcntr}

\usepackage{tikz}
\usepackage{float}
\usepackage{cancel}
\usepackage{algorithm}
\usepackage{algpseudocode}

\usetikzlibrary{arrows,automata}

\def\N{\mathbb{N}}

\def\ta{\mathtt{a}}
\def\tb{\mathtt{b}}
\def\tc{\mathtt{c}}

\DeclareMathOperator{\Subseq}{SubSeq}

\DeclareMathOperator{\letters}{alph}

\DeclareMathOperator{\ar}{ar}

\def\r{\operatorname{r}}

\DeclareMathOperator{\labels}{label}
\newcommand{\poly}{poly}
\DeclareMathOperator{\al}{alph}

\def\nth#1{#1$^{\text{th}}$}

\renewcommand{\epsilon}{\varepsilon}

\tikzstyle{edgeLabel}=[inner sep=0.5mm,fill=white,text=black]

\newif\ifpaper
\paperfalse 

\begin{document}
	\title{$k$-Universality of Regular Languages Revisited}
	%
	%
	\author{Duncan Adamson\inst{1}\orcidID{0000-0003-3343-2435} \and
		Pamela Fleischmann\inst{2}\orcidID{0000-0002-1531-7970} \and
		Annika Huch\inst{2}\orcidID{0009-0005-1145-5806} \and
		Tore Koß\inst{3}\orcidID{0000-0001-6002-1581} \and
		Florin Manea\inst{3}\orcidID{0000-0001-6094-3324}}
	\authorrunning{D. Adamson et al.}
	%
	\institute{School of Computer Science, University of St Andrews, UK \email{duncan.adamson@st-andrews.ac.uk} \and
		Department of Computer Science, Kiel University, Germany
		\email{\{fpa,ahu\}@informatik.uni-kiel.de} \and
		Department of Computer Science, University of Göttingen, Germany\\
		\email{\{tore.koss,florin.manea\}@cs.uni-goettingen.de}}
	\maketitle              
	\begin{abstract}
		\input{abstract}
		\keywords{String Algorithms \and Regular Languages \and Subregular Languages \and Regular Expressions \and Finite Automata \and Subsequences \and Universality}
	\end{abstract}
	\newpage
	
	\setcounter{page}{1}
	\section{Introduction}
	\input{introduction}

	\section{Preliminaries}\label{prelims}
	\input{prelims}

\input{general}
	\label{sec:acyclic}
	\input{acyclic}
	\label{sec:regex}
	\input{regex}
	
	
	

	%
	%
	
%
	
	\newpage
	%
	%
	
	 \bibliographystyle{splncs04}
	 \bibliography{refs}

\end{document}

%% file: abstract.tex
A subsequence of a word $w$ is a word $u$ such that $u = w[i_1] w[i_2] \cdots w[i_k]$, for some set of indices $1 \leq i_1 < i_2 < \dots < i_k \leq \vert w \vert$. A word $w$ is \emph{$k$-subsequence universal} over an alphabet $\Sigma$ if every word over $\Sigma$ up to length $k$ appears in $w$ as a subsequence. In this paper, we revisit the problem $k$-ESU of deciding, for a given integer $k$, whether a regular language, given either as nondeterministic finite automaton or as a regular expression, contains a $k$-universal word.  [Adamson et al., ISAAC 2023] showed that this problem is NP-hard, even in the case when $k=1$, and an FPT algorithm w.r.t. the size of the input alphabet was given. In this paper, we improve the aforementioned algorithmic result and complete the analysis of this problem w.r.t. other parameters. That is, we propose a more efficient FPT algorithm for $k$-ESU, with respect to the size of the input alphabet, and propose new FPT algorithms for this problem w.r.t.~the number of states of the input automaton and the length of the input regular expression. We also discuss corresponding lower bounds. Our results significantly improve the understanding of this problem. 

%

%% file: introduction.tex
Words and subsequences are two fundamental and heavily studied concepts in the field of combinatorics on words. A subsequence of a given word $w$ is a word $u$ that is obtained by deleting some letters of $w$ while preserving the order of letters in $w$ for the non-deleted ones. 
For instance, the word $\mathtt{subregularuniversality}$ has the scattered factors $\mathtt{glue}, \mathtt{guilty}, \mathtt{surreality}$ while $\mathtt{versailles}$ is not a scattered factor since the letters do not occur in the same order as in the whole word. \looseness=-1

Within theoretical computer science, subsequences were considered mainly in two areas. On the one hand, there are many algorithmic and complexity problems arising from this object, e.g., in problems such as the longest or shortest common subsequence or supersequence problems \cite{chvatal1975longest,Hirschberg77,HuntS77,Maier78,MasekP80,NakatsuKY82,Baeza-Yates91,BergrothHR00,BringmannC18,BringmannK18,AbboudEtAl2014,AbboudEtAl2015,AbboudRubinstein2018}, or matching and analysis problems related to the sets of subsequences occurring in a word \cite{hebrard1991algorithm,simon2003words,tronicek2003common,DBLP:journals/jda/CrochemoreMT03,barker2020scattered,day2021edit,fleischer2018testing,KimKH22,gawrychowski2021simons,kosche2021absent,DBLP:conf/rp/FleischmannKKMNSW23}; see also the survey \cite{Kosche2022SubsequenceSurvey}. On the other hand, subsequences are intensely studied objects in combinatorics on words, formal languages, automata theory and logics, especially with respect to their strong relation to piecewise-testable languages \cite{KarandikarKS15,karandikar2016height,Simon72,simon1972thesis}, or subword orders and downward closures \cite{halfon2017decidability,DBLP:conf/csr/Kuske20,DBLP:conf/fossacs/KuskeZ19,zetzsche2016complexity,DBLP:conf/lics/Zetzsche18}.
Such theoretical results found applications in a wide number of fields including bioinformatics \cite{han2020novel,shikder2019openmp}, and modelling concurrency \cite{shaw1978software}. Still on the applicative side, a new research direction, originating in database theory \cite{artikis2017complex,Kleest-Meissner22,Kleest-Meissner23,FrochauxK23}, deals with constrained-subsequences of strings, where the substrings occurring between the positions of the subsequence are subject to regular or length constraints; a series of algorithmic and complexity results related to this setting were obtained \cite{DayKMS22,KoscheKMP22,Goettingen2023words,ManeaRS24}. \looseness=-1

The work of this paper is positioned in between the study of algorithmic properties of subsequences and automata and formal language theory. More precisely, we revisit and extend the results of \cite{SchnoebelenV23,regunivpaper,fazekas2024subsequencematchinganalysisproblems}. In the respective works, the authors generalized fundamental algorithmic questions related to the subsequences occurring in a single word to the case of subsequences occurring in (finite or infinite) sets of words represented succinctly by automata accepting these sets or by grammars generating them (i.e., formal languages). In general, two main classes of questions were investigated: for some subsequence-related property of words, given a generative or accepting mechanism for a language, one is interested, on the one hand, whether there is at least one word of the language which fulfils the respective property, or, on the other hand, whether all words of the language fulfil the respective property. The main such subsequence-property approached in \cite{SchnoebelenV23,regunivpaper,fazekas2024subsequencematchinganalysisproblems} is the $k$-universality of words. Formally, a word over an alphabet $\Sigma$ is called \emph{$k$-universal} if its set of subsequences of length $k$ includes every word of length $k$ over $\Sigma$. An extensive literature deals with this concept, mostly from a theoretical perspective \cite{karandikar2016height,schnoebelen2019height,adamson2023words,barker2020scattered,day2021edit,fleischmann2021scattered,fleischmann2023alphabetafactorization,kosche2021absent,SchnoebelenV23}; see \cite{day2021edit,regunivpaper} and the references therein for a detailed discussion of the motivation for the study of $k$-universality. Coming back to our focus, in \cite{regunivpaper}, the authors investigated two decision problems regarding universality, corresponding to the two directions mentioned above. First is the problem of deciding whether there exists a $k$-universal word in a given regular language $L$, specified as a non-deterministic finite automaton, for some given $k$ (for short $k$-ESU for $k$ \textbf{e}xistence \textbf{s}ubsequence \textbf{u}niversality); this problem was shown to be NP-complete, and an FPT algorithm, with respect to the size of the input alphabet, was given for it. Second is the problem whether all words in $L$ are $k$-universal ($k$-USU for $k$ \textbf{u}niversal \textbf{s}ubsequence \textbf{u}niversality); this problem was shown to be solvable in polynomial time. In \cite{fazekas2024subsequencematchinganalysisproblems}, the authors extended the study of these problems for context-free and context-sensitive languages, specified by grammars generating them;  \cite{SchnoebelenV23} discusses similar problems for strings encoded via straight-line programs. In this context, the problem $k$-ESU is the harder and less understood of the two aforementioned problems, so we revisit it in this paper. \looseness=-1

The research direction proposed in \cite{SchnoebelenV23,regunivpaper,fazekas2024subsequencematchinganalysisproblems}, and revisited here, is not completely new. It originates in the famous result of Higman~\cite{higman1952ordering} which states that the downward closure of every language (i.e., the set of all subsequences of the strings of the respective language) is regular; however, an automaton accepting this language is not always computable (e.g., for the class of context-sensitive languages). As downward closures of languages can be seen as structurally-simple yet faithful-enough abstractions of complex (and practically relevant) formal languages, they become useful in practical applications and well studied from a theoretical point of view (see~\cite{zetzsche2016complexity,DBLP:conf/lics/Zetzsche18,AnandZ23} and the references therein). In our setting, downward closures, the fact that they are regular and it is not hard to compute an automaton for them (since we deal with regular languages only), could also be useful in approaching problems such as $k$-ESU. But, as discussed in \cite{fazekas2024subsequencematchinganalysisproblems}, such solutions would be quite inefficient. So, our main aim in revisiting $k$-ESU is to extend the work of \cite{regunivpaper,fazekas2024subsequencematchinganalysisproblems} on it, and identify new, more efficient solutions, as well as a better understanding of the lower bounds for solving this problem. Moreover, as $k$-ESU was already shown to be hard for regular languages (even for finite languages), we also consider first this class of languages in this paper, and leave the investigation of more complex classes of languages as future work.\looseness=-1

%
%
%
%
%
%
%
%

\emph{Our Contributions.} Our main results provide a very detailed image of the parametrised complexity of the $k$-ESU problem for regular languages. When the input regular language is given as an NFA, this problem has three parameters: the number $n$ of states of the NFA, the size $\sigma$ of the input alphabet, and the integer $k$. We analyse the complexity of $k$-ESU with respect to all these parameters. As mentioned above, an FPT algorithm w.r.t. $\sigma$ was known from \cite{regunivpaper}; we give here a faster algorithm, which runs in linear time for constant alphabets. We also give here the first FPT algorithm for $k$-ESU w.r.t. $n$, which runs in linear time for NFAs with a constant number of states. As already $1$-ESU is NP-hard (when the input regular language is given as NFA, DFA, or regular expression), we can conclude that, unless P=NP, there is no FPT algorithm, w.r.t. the parameter $k$, solving $k$-ESU. When the input regular language is given as a regular expression of length $n$ over an alphabet of size $\sigma$, we obtain an FPT algorithm w.r.t. $\sigma$; $1$-ESU remains NP-complete in this setting. Finally, based on the lower bounds shown here and in \cite{regunivpaper}, we show that our algorithmic results are tight, from a fine-grained complexity perspective. As far as the used techniques are concerned, we extensively use a graph-theoretical tool box, which allows us to gain new combinatorial insights in the structure of the $k$-ESU problem, and develop efficient methods for solving it. \looseness=-1

%% file: prelims.tex
Let  $\N = \{1,2,\ldots\}$ denote the natural numbers and set $\N_0 = \N
\cup \{0\}$. Let $[n]=\{1,\ldots,n\}$ and let $[i,n]=\{i, i+1, \ldots, n\}$ for all $i,n\in\N_0$ with $i \leq n$.

An \emph{alphabet} $\Sigma=\{1,2,\ldots,\sigma\}$ is a finite set of symbols, called
\emph{letters} (w.l.o.g., we can assume that the letters are integers). A \emph{word} (also known as a \emph{string}) $w$ is a finite sequence of letters from a given alphabet. The length of a word $w$, denoted $\vert w \vert$ is the number of letters in the word. For $i \in
[|w|]$ let $w[i]$ denote the $i^{th}$ letter of $w$.  The set of all
finite words over the alphabet $\Sigma$, denoted by $\Sigma^{\ast}$, is the free monoid generated by $\Sigma$ with
concatenation as operation and the neutral element is the empty word
$\varepsilon$, i.e., the word of length $0$. Given $n \in \N_0$, let $\Sigma^n$ denote all words in $\Sigma^{\ast}$ exactly of length $n$ and $\Sigma^{\leq n}$ the set of all words of $\Sigma^{\ast}$ of length at most $n$.
Let $\letters(w) = \{\ta \in \Sigma \mid \exists i \in [|w|]: w[i] = \ta \}$ be the alphabet of $w$. 
For $u,w\in\Sigma^{\ast}$, $u$ is called a \emph{factor}
of $w$, if $w = xuy$ for some words $x,y\in\Sigma^{\ast}$. If $x = \varepsilon$ (resp.,
$y = \varepsilon$) then $u$ is called a \emph{prefix} (resp., \emph{suffix}) of
$w$.  For $1\leq i\leq j\leq|w|$ 
define the factor from $w$'s \nth{$i$} letter to the \nth{$j$} letter by  $w[i,j]=w[i]\cdots w[j]$.
Given a pair of indices $i  < j$, we assume $w[j, i] = \varepsilon$.
 For $w\in\Sigma^{\ast}$ and $n\in\N_0$ define inductively $w^0=\varepsilon$ and $w^n=ww^{n-1}$.



As we are interested in investigating the $k$-subsequence universality of subregular languages, we first introduce the basic concepts related to subsequences. We then present the definitions for the transformation of these notions to the domain of subregular languages, finite automata and other language models.

\begin{definition}
	Let $w \in \Sigma^*$ and $n \in \N_0$. A word $u \in \Sigma^*$ is called \emph{subsequence} of $w$ ($u \in \Subseq(w)$) if there exist $v_1, \ldots, v_{n+1} \in \Sigma^*$ such that $w = v_1 u[1] v_2 u[2] \cdots v_n u[n] v_{n+1}$. Let $\Subseq_k(w) = \{u \in \Subseq(w) \mid \vert u \vert = k\}$.
\end{definition}

	Subsequences of $\mathtt{automatauniversality}$ are $\mathtt{auto}$, $\mathtt{tomata}$, $\mathtt{salty}$,  and $\mathtt{atom}$ while $\mathtt{star}$, $\mathtt{alien}$ are not because their letters occur in the wrong order.

In \cite{barker2020scattered}, the authors investigated words which have, for a given $k\in\N_0$, all words from $\Sigma^{k}$ as subsequence, namely $k$-subsequence universal words. Note that this notion is similar to the one of richness introduced and investigated in \cite{KarandikarKS15,karandikar2016height}. We stick here to the notion of $k$-subsequence universality since our focus is subregular languages and thus the well-known notion of the universality of automata and formal languages, i.e., $L({A})=\Sigma^{\ast}$ for a given finite automaton ${A}$, is close to the one of subsequence universality of words.

\begin{definition}\label{def:subsequence}
	A word $w \in \Sigma^*$ is called \emph{$k$-subsequence universal} (w.r.t. $\Sigma$), for $k \in \N_0$, if	
	$\Subseq_k(w) = \Sigma^k$. If the context is clear we call $w$ $k$-universal. The universality-index $\iota(w)$  is the largest $k$ such that $w$ is $k$-universal.
\end{definition}


Further, we recall the {\em arch factorisation} by Hébrard \cite{hebrard1991algorithm}.

\begin{definition}
The {\em arch factorisation} of $w \in \Sigma^*$ is defined by $w = \ar_1(w) \cdots$ $\ar_k(w)$ $\r(w)$ for $k \in \mathbb{N}_0$ with
$\ar_i(w)[\vert \ar_i(w) \vert] \notin \letters(\ar_i(w)[1, \vert \ar_i(w) \vert - 1 ])$ for all $i \in [k]$ and $\iota(\ar_i(w))=1$, as well as $\letters(\r(w)) \subsetneq \Sigma$.
	The words $\ar_i(w)$ are the \emph{arches} and $\r(w)$ is the \emph{rest} of $w$.
\end{definition}

\begin{example}\label{prelims:exampleuniversalitycondt}
\label{prelims:exampleuniversality}
	Consider $w=\mathtt{baaababb} \in \{\ta,\tb\}^*$. We have $\vert \Subseq_3 (\mathtt{baaababb})\vert = \vert \{\mathtt{aaa}, \mathtt{aab}, \mathtt{aba}, \mathtt{abb}, \mathtt{baa}, \mathtt{bab}, \mathtt{bba}, \mathtt{bbb}\} \vert = 2^3$. Since $\mathtt{abba} \notin \Subseq_4(\mathtt{baaababb})$, it follows that $\iota(\mathtt{baaababb}) = 3$.
Further, $w=\mathtt{(ba) \cdot (aab) \cdot (ab) \cdot b}$ is the arch factorisation of $w$ where the parentheses denote the three arches and the rest $\mathtt{b}$.
\end{example}
%
%


\begin{theorem}[\cite{barker2020scattered}]
	Let $w \in \Sigma^{\geq k}$ with $\al(w) = \Sigma$. Then we have $\iota(w)=k$ iff $w$ has exactly $k$ arches.
\end{theorem}

We now recall basic notions on finite automata, for further details we refer to \cite{HopcroftU79}. 
%
%
%
A {\em non-deterministic finite automaton (NFA)} ${A}$ is a tuple $(Q,\Sigma,q_0,\delta,F)$ with the finite set of states $Q$ (of cardinality $n\in\N$), an initial state $q_0\in Q$, the set of final states $F\subseteq Q$, an input alphabet $\Sigma$, and a transition function $\delta:Q\times (\Sigma \cup \{\epsilon\}) \rightarrow 2^{Q}$, where $2^{Q}$ is the powerset of $Q$. If $q_2\in \delta(q_1,a)$, for some $a\in \Sigma \cup \{\epsilon\}$, then we have the transition $(q_1,a,q_2)$ in $A$. If we have $|\delta(q,a)|= 1$ and $\delta(q,\epsilon)= \emptyset$ for all $q\in Q, a\in \Sigma$, then ${A}$ is called {\em deterministic} (DFA). 

We call a sequence of transitions $\beta=(q,a_1,q_1)(q_1,a_2,q_2)\ldots(q_{\ell-1},a_{\ell},q_\ell)$ an \emph{$\ell$-length walk} from any $q \in Q$ to $q_\ell$ in ${A}$; in such a walk, we have $q_i\in Q$ for all $i\in[0,\ell]$, $a_i\in \Sigma\cup\{\epsilon\}$, and $q_{i+1}\in \delta(q_i,\ta_{i+1})$, for all $i\in[0,\ell-1]$.  The word $a_1\cdots a_\ell$ is the {\em label of $\beta$} and is denoted by $\labels(\beta)$. A state $q\in Q$ is called \emph{accessible} (respectively,\emph{co-accessible}) in ${A}$ if there exists a walk connecting $q_0$ to $q$ (respectively, $q$ to a final state). A walk is called {\em accepting} if $q_{\ell}\in F$ holds. The language of ${A}$, i.e., the set of words {\em accepted} by ${A}$, is $L({A})=\{w\in\Sigma^{\ast}\mid \exists \mbox{ accepting walk }\beta\mbox{ in }{A}:\,w=\labels(\beta)\}$. 
Note that the class of languages accepted by NFAs is equal to the class of languages accepted by DFAs and it is equal to the class of regular languages. Moreover, for every word $w\in L({A})$ there exists exactly one (in the deterministic case) or a set (in the non-deterministic case) of walk(s) labelled with $w$. \looseness=-1
 

 Further, we define regular expressions (regex, for short) and their semantics.

\begin{definition}
	Let $\Sigma$ be an alphabet. Then $\emptyset, \varepsilon,$ and $\ta \in \Sigma$ are \emph{regular expressions}.
	Further, for two regular expressions $R_1$ and $R_2$, the terms $R_1 R_2$, $(R_1\mid R_2)$ and $(R_1)^*$ are regular expressions. 
	
	A regular expression $R$ that contains no $*$ is called \emph{star-free}.

	Let $R$ be a regular expression. Then the semantics of $R$ is inductively defined by the base cases $L (\emptyset) = \emptyset$, $L(\varepsilon) = \{\varepsilon\}$ and ${L}(\ta) = \{\ta\}$ for $\ta\in\Sigma$.
	Further, we define $ L(R_1 R_2) =  L(R_1)  L(R_2)$, $ L((R_1 | R_2)) =  L(R_1)\cup  L(R_2)$ and $ L((R_1)^*) =  L(R_1)^*$ for regular expressions $R_1,R_2$.
	
	$ L(R)$ is called the language described by the regular expression $R$. 
\end{definition}

 For example, $ L (\ta^*\tb|\varepsilon) =  L (\ta^* \tb) \cup \{\varepsilon\} =  L (\ta^*) \{\tb\} \cup \{\varepsilon\} = \{\ta\}^* \{\tb\} \cup \{\varepsilon\}$.

Now, we can define the $k$-subsequence universality of a regular language~$L$. Here, we distinguish whether at least one or all words of $L$ are $k$-universal w.r.t.~Definition~\ref{def:subsequence}. Note that we always take the minimal alphabet $\Sigma$ such that $L\subseteq \Sigma^*$ as a reference when considering the $k$-universality of words from $L$, as otherwise, for larger alphabets, the universality is trivially $0$.

\begin{definition}\label{def:univ} Let $k\in \N$ be an integer.
A language $L$ is {\em existence $k$-subsequence universal} ($k$-$\exists$-universal) if there
exists a $k$-universal word $w\in L$, and $L$ is {\em universal $k$-subsequence universal} ($k$-$\forall$-universal) if all words $w\in L$ are $k$-universal. \looseness=-1
\end{definition}

If a language $L$, accepted (represented) by some automaton ${A}$ (or regular expression $R$), is $k$-$\exists$-universal (respectively, $k$-$\forall$-universal) then we also say that the automaton ${A}$ (regular expression $R$) is $k$-$\exists$-universal (respectively, $k$-$\forall$-universal). We also say that a {\em walk $\pi$ in ${A}$ is $k$-universal} if $\labels(\pi)$ is $k$-universal. 

For instance, the expression $(\ta^*\tb^*\tc | \tb\tc) | \ta^*$ is a $1$-$\exists$-universal regular expression that is not $2$-$\exists$-universal. Further examples are given in Figure~\ref{img:universal_automaton}. 

Starting from Definition \ref{def:univ}, two decision problems were introduced in \cite{regunivpaper}. 

\ifpaper
\else
\input{figure}
\fi

\begin{itemize}[leftmargin=10pt,itemsep=1cm,parsep=-1cm]
\item[1.] The {\em existence subsequence universality problem for regular languages ($k$-ESU)} is to decide for a language, given by an NFA $\mathcal{A}$ recognising it (respectively, by a regex $R$ describing it), and $k\in\N$ whether $L$ is $k$-$\exists$-universal. 

\item[2.] The {\em universal subsequence universality problem for regular languages ($k$-USU)} is to decide for a language, given by an NFA $\mathcal{A}$ recognising it (regex $R$ describing it), and $k\in\N$ whether $L$ is $k$-$\forall$-universal.
		\vspace*{-5pt}
\end{itemize}

Here we focus on the first problem, as $k$-USU can already be solved in polynomial time for all regular languages, given as NFAs or regexes.

\ifpaper
\else
\subsubsection*{Computational Model}
\input{compmodel}
\fi

%% file: figure.tex
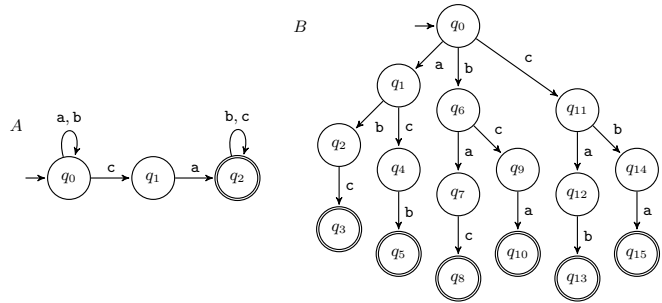
\begin{figure}
	\centering
	\begin{minipage}[c]{0.3\textwidth}
		\scalebox{0.7}
		{\begin{tikzpicture}
				[->,>=stealth',shorten >=1pt,auto,node distance=1.6cm,semithick]
				\tikzstyle{every state}=[initial text=$ $]
				\node[initial,state] (A)                    {$q_0$};
				\node[state]         (B) [right of=A]       {$q_1$};
				\node[state,accepting] (D) [right of=B]     {$q_2$};
				
				\path (A) edge [loop above] node {$\mathtt{a}, \mathtt{b}$} (A)
				edge              node {$\mathtt{c}$} (B)
				(B) edge              node {$\mathtt{a}$} (D)
				(D) edge [loop above] node {$\mathtt{b,c}$} (D);
				\node at (-1,1) {${A}$};
		\end{tikzpicture}}
	\end{minipage}	
	\begin{minipage}[c]{0.4\textwidth}
		\scalebox{0.7}
		{\begin{tikzpicture}[->,>=stealth',shorten >=1pt,auto,node distance=1.6cm,
				semithick]
				\tikzstyle{every state}=[initial text=$ $]
				\node[initial,state] (A)                    {$q_0$};
				\node[state]         (B) [below left of=A]       {$q_1$};
				\node[state]         (C) [below left of=B]       {$q_2$};
				\node[state,accepting] (D) [below of=C]       {$q_3$};
				\node[state]         (E) [below of=B]       {$q_4$};
				\node[state,accepting] (F) [below of=E]       {$q_5$};
				\node[state]         (q6) [below of=A]        {$q_6$};
				\node[state]         (q7) [below of=q6]       {$q_7$};
				\node[state,accepting] (q8) [below of=q7]       {$q_8$};
				\node[state]         (q9) [below right of=q6]       {$q_9$};
				\node[state,accepting] (q10) [below of=q9]       {$q_{10}$};
				\node[state]         (q11) [above right of=q9]       {$q_{11}$};
				\node[state]         (q12) [below of=q11]       {$q_{12}$};
				\node[state,accepting] (q13) [below of=q12]       {$q_{13}$};
				\node[state]         (q14) [below right of=q11]       {$q_{14}$};
				\node[state,accepting] (q15) [below of=q14]       {$q_{15}$};
				
				\path (A) edge node {$\mathtt{a}$} (B)
				edge node {$\mathtt{b}$} (q6)
				edge node {$\mathtt{c}$} (q11)
				(B)	edge node {$\mathtt{b}$} (C)
				edge node {$\mathtt{c}$} (E)
				(C) edge node {$\mathtt{c}$} (D)
				(E) edge node {$\mathtt{b}$} (F)
				(q6) edge node {$\mathtt{a}$} (q7)
				edge node {$\mathtt{c}$} (q9) 
				(q7) edge node {$\mathtt{c}$} (q8)
				(q9) edge node {$\mathtt{a}$} (q10)
				(q11) edge node {$\mathtt{a}$} (q12)
				edge node {$\mathtt{b}$} (q14)
				(q12) edge node {$\mathtt{b}$} (q13)
				(q14) edge node {$\mathtt{a}$} (q15);
				\node at (-3,0) {${B}$};
		\end{tikzpicture}}
	\end{minipage}
	\caption{A $2$-$\exists$-universal NFA ${A}$ and a $1$-$\forall$-universal  NFA ${B}$. In ${A}$, we have the 2 universal word $\mathtt{a}\mathtt{b}\mathtt{c}\mathtt{a}\mathtt{b}\mathtt{c}$, however no word can contain more than $1$ occurrence of $\mathtt{a}$ after the first occurrence of $\mathtt{c}$, thus can contain at most $2$-arches. In ${B}$, we have a regular language that recognises the set of permutations of $\mathtt{a} \mathtt{b} \mathtt{c}$, thus each word contains exactly one arch and hence is $1$-universal.}
	\label{img:universal_automaton}
\end{figure}

%% file: compmodel.tex
The computational model we use to state our algorithms is the standard unit-cost word RAM with logarithmic word-size $\omega$ (meaning that each memory word can hold $\omega$ bits). It is assumed that this model allows processing inputs of size $n$, where $\omega \geq \log n$; in other words, the size $n$ of the data never exceeds (but, in the worst case, is equal to) $2^\omega$. Intuitively, the size of the memory word is determined by the processor, and larger inputs require a stronger processor (which can, of course, deal with much smaller inputs as well). Indirect addressing and basic arithmetical operations on such memory words are assumed to work in constant time. Note that numbers with $\ell$ bits are represented in $\O(\ell/\omega )$ memory words, and working with them takes time proportional to the number of memory words on which they are represented. This is a standard computational model for the analysis of algorithms, defined in \cite{FredmanW90}. Our algorithms have automata, strings (regular expressions), and numbers as input, so we follow standard assumptions, Namely, we work with languages over {\em integer alphabets} (see, e.\,g.,~\cite{crochemore}): whenever we are given an input of size $n$, we assume that the alphabet of the given automata or regex is $\Sigma=\{1,2,\ldots,\sigma\}$, with $|\Sigma|=\sigma\leq n$; we also assume that the states of the input automata, when we deal with this model, are given as integers from $\{1,\ldots,q\}$, where $q\leq n$. Also, automata are specified as a list of states (with the initial state and the final ones highlighted), and, for each state we are given the list of transitions leaving it. 

One of our algorithms (from Theorem \ref{thm:FPT_sigma}, which is then used as a subroutine in Theorem \ref{thm:FPT_sigma_regex}, as well) uses exponential space w.r.t.\ the size of the input alphabet, namely $\sigma$, but polynomial w.r.t.\ all the other components of the input. To avoid clutter, we assume that our exponential-time and -space algorithms runs on a RAM model where we can allocate as much memory as our algorithms needs (i.e., the size of the memory-word $\omega$ is big enough to allow addressing all the memory we need in this algorithm in constant time). For the case of $\sigma \in O(1)$, this additional assumption becomes superfluous; for non-constant $\sigma$, we emphasise that the big size of memory words is only used for building and addressing large data structures (that is, large matrices), but not for speeding up our algorithms by, e.g., allowing constant-time operations on big numbers (that is, numbers exponential in the size of the input).



%% file: general.tex
\section{Preprocessing: Graph-theoretic View on NFAs}\label{sec:preProc}

In the problems discussed in this paper, we consider a regular language $L$, over the alphabet $\Sigma$, with $|\Sigma|=\sigma$, given either as a regular expression or as an NFA. In this section we only consider the case when the language is given as an NFA $A=(Q,\Sigma,q_0,F,\delta)$ with $n$ states and $m$ transitions, and initial state $q_0$. 

We aim to obtain efficient algorithms that, for a given $k$, check whether $L$ contains a word $w$ of universality $\iota(w)\geq k$. In particular, we will discuss two algorithms, one which is FPT with respect to the parameter $n$, and one which is FPT w.r.t. the parameter $\sigma$. In this section, we describe a series of steps that will be executed as preprocessing for both these algorithms, and introduce a series of definitions used in the description of these algorithms.

Firstly, we note that we can assume w.l.o.g. that all states in $A$ are both accessible and co-accessible and that $A$ has a single final state $f$ (i.e, $F=\{f\}$), which is not the origin of any transition (i.e., $\delta(f,a)=\emptyset$ for all $a\in \Sigma$). Indeed, if we are given an arbitrary NFA $A'$, we can transform it, by standard methods (see, e.g., \cite{HopcroftU79}), in linear time in the total size of $A'$ (that is, number of states plus number of transitions), and obtain an automaton $A$ for which these assumptions hold. Clearly, for $A$, we also have that $\sigma\leq m\leq n^2 \sigma$ (as for each pair of states $(q_1,q_2)$, we can have at most $\sigma$ transitions from $q_1$ to $q_2$, one for each letter, and, for each letter, there should be at least one transition labelled with that letter). 

We now make several remarks on NFAs that are useful for our results. 

Each NFA $A$ can be seen canonically as a directed graph $G_A$, whose vertex-set is the set $Q$ of states of $A$, and there is an edge from state $q_1$ to state $q_2$ in the graph $G_A$ if an only if there is a transition from $q_1$ to $q_2$ in $A$. We recall the following lemma, see, e.g. \cite{Tarjan72}. \looseness=-1

\begin{lemma}\label{lem:stronglyConnected}
Given an NFA $A$, we can compute in $O(n+m)$ time the decomposition of the associated graph $G_A$ in strongly connected components ${\mathcal C}_1, \ldots, {\mathcal C}_e$.
\end{lemma}
The complexity of this algorithm is determined by the computation of the graph $G_A$ from $A$, which can be done in $O(n+m)$ time, and then the usage of Tarjan's algorithm \cite{Tarjan72} for computing the strongly connected components ${\mathcal C}_1, \ldots, {\mathcal C}_e$ which also requires $O(n+m)$ time. We can also compute, within the same time complexity, the directed acyclic graph $S(A)$ of the strongly connected components of $G_A$: it has the vertices $\{1,\ldots,e\}$, corresponding to the strongly connected components ${\mathcal C}_1, \ldots, {\mathcal C}_e$ of $G$, respectively, and there is an edge from $i$ to $j$, for $i,j\in [e]$, if and only if there is an edge from a vertex of ${\mathcal C}_i$ to some vertex of ${\mathcal C}_j$. 

 We can assume that each strongly connected component produced in this algorithm is output as a list of vertices, and, for each vertex, the list of the edges leaving that vertex. For a state $q\in Q$, we set $S(q)=i$ if ${\mathcal C}_i$ is the strongly connected component of $G_A$ which contains state~$q$. 
 
The decomposition of $G_A$ in strongly connected components canonically corresponds to a decomposition of the NFA $A$ in strongly connected components: for $i\in [e]$, the component ${\mathcal C}_i$ of $A$ contains the states corresponding to the vertices of the component ${\mathcal C}_i$ of $G_A$, as well as all the transitions between them. The main difference between the components of $A$ and those of $G_A$ is that in a component of $A$ we might have multiple transitions between two states, while in the respective component of $G_A$ there is at most one edge from one vertex to another. Clearly, we can compute this decomposition of $A$ in strongly connected components in $O(n+m)$ time.

This concludes the description of the preprocessing steps performed at the beginning of our algorithms. 

Other graph-theoretic notions can be transferred canonically from $G_A$ to~$A$. Recall that a walk in the NFA $A$, originating in state $q$ and with target $q_n$, is a sequence of transitions $\beta=(q,a_1,q_1)(q_1,a_2,q_2)\cdots (q_{n-1},a_n,q_n)$, where  $q_i\in Q$, $a_i\in \Sigma \cup\{\epsilon\}$ for $i\in [n]$, and $q_i\in \delta(q_{i-1},a_i)$ for $i\geq 2$, and $q_1\in \delta(q,a_1)$. The label of $\beta$ is the word $\labels(\beta)=a_1\cdots a_n$. If $q=q_n$, then $\beta$ is called a cycle.

It is not hard to see that, for each state $q$, there is a cycle, denoted in the following $c_q$, containing $q$ that goes at least once through every transition of the NFA-component ${\mathcal C}_{S(q)}$. Indeed, for each transition $(q_1,a,q_2)$, where both $q_1$ and $q_2$ are in ${\mathcal C}_{S(q)}$, we can follow the walk connecting $q$ to $q_1$, then the transition $(q_1,a,q_2)$, and then follow the walk from $q_2$ to $q$; we then simply take the concatenation of all these cycles, and get $c_q$. So, in the following, for each state $q$, we will use $c_q$ to denote a cycle which goes through $q$ and traverses at least once every transition between states contained in ${\mathcal C}_{S(q)}$; we emphasise already at this point that this $c_q$ is never used in our algorithmic results, so we do not focus on constructing it efficiently, but, rather, it is important to know that such a cycle exists for each~$q$. Also, note that even if $S(q_1)=S(q_2)$, we do not necessarily have that $c_{q_1}=c_{q_2}$, but we do have that $\al(\labels(c_{q_1}))= \al(\labels(c_{q_2}))$. For states $q,q'$ with $S(q)=S(q')$, let $p_{q,q'}$ denote the shortest prefix of $c_q$ (w.r.t. number of transitions) connecting state $q$ to state $q'$ in the NFA $A$. 

A walk expression is:
\begin{itemize}[leftmargin=10pt,itemsep=1cm,parsep=-1cm]
\item a {\em single transition} $(q_1,a_2,q_2)$. This denotes a single one-edge walk, namely $(q_1,a_2,q_2)$. The origin of this walk expression is $q_1$ and the target $q_2$. Transitions are called atomic walk expressions (atoms, for short). 
\item an {\em extended cycle}, written as $(q[c_q]q',b,q'')$, where $q,q',q''\in Q$, $q'$ is a state of $\mathcal C_{S(q)}$, $b\in\Sigma\cup\{\epsilon\}$, and $q''\in \delta(q',b)$. This denotes a set of walks, namely those walks $c_q^i p_{q,q'} (q',b,q'')$ obtained by going $i$ times around the cycle $c_q$ (i.e., following $i$ times the transitions of the cycle $c_q$, starting in $q$), for some $i\geq 0$, then following the path $p_{q,q'}$, and, finally the transition $(q',b,q'')$, which corresponds to $q''\in \delta(q',b)$. The origin of this walk expression is $q$ and the target $q''$. Such walk expressions are also called atomic (or atoms). 
\item a concatenation $\alpha= \alpha_1 \alpha_2$ of walk expressions, such that the origin of $\alpha_2$ is the same as the target of $\alpha_1$. This expression denotes the set containing all walks obtained by concatenating walks denoted by $\alpha_1$ with walks denoted by $\alpha_2$. The origin of $\alpha$ is the origin of $\alpha_1$, the target of $\alpha$ is the target of $\alpha_2$. These are non atomic walk expressions.
		\vspace*{-5pt}
\end{itemize}

\section{FPT-Algorithms for $k$-ESU}\label{sec:FPT}

We assume that we are in the setting defined in Section \ref{sec:preProc}, and use the notations introduced in that section. We begin with a simple observation.
\begin{lemma}\label{lem:infUniv} $L$ contains, for each $i\in \N$, a word $w_i$ of universality $\iota(w_i)\geq i$ if and only if there exists a state $q$ of $A$ such that $\al(\labels(c_q))=\Sigma$. 
\end{lemma}
\ifpaper
\else
\input{proofs/proof_infUniv}
\fi
It is not hard to see that if $L$ does not contain words or arbitrarily large universality then $\iota(w)\leq n$ for all $w\in L$.

The following lemma follows immediately from the preprocessing described in Section \ref{sec:preProc} and shows that we can test this property efficiently.
\begin{lemma}\label{lem:cycles} We can compute in $O(m+n)$ time the sets $V_i\subseteq \Sigma\cup \{\epsilon\}$ of all labels of transitions of the strongly connected component ${\mathcal C}_i$ of $A$, for $i\in [e]$. We can then access in $O(1)$ the set $V_{S(q)}$ with $\al(\labels(c_q))=V_{S(q)}$, for all $q\in Q$.
\end{lemma}
\ifpaper
\else
\input{proofs/proof_cycles}
\fi

Therefore, as a first main consequence of the previous lemma, we can test efficiently the existence in $A$ of a state $q$ such that $\al(\labels(c_q))=V_{S(q)}=\Sigma$, as in Lemma \ref{lem:infUniv}. If such a state exists, then we can trivially answer to $k$-ESU, for the input NFA $A$, positively, in polynomial time (in the size of $A$). 

So, let us assume, from now on, that there is no state $q$ with $\al(\labels(c_q))=\Sigma$. We emphasise that Lemma \ref{lem:cycles} allows us to compute in polynomial time, for each state $q\in Q$, only the set $V_{S(q)}= \al(\labels(c_q))$; as already mentioned in the previous section, we do not compute effectively the cycle $c_q$ (we just collect for each $q$ the letters which label all the transitions in its strongly connected component).  However, in the following, in our combinatorial proofs, we will still make use of the notation $c_q$, as introduced above. Since there is no transition leaving the final state $f$ of $A$, $V_{S(f)}=\emptyset$ and $c_f$ is the empty walk.

We will now move on and describe our main results: the first FPT algorithm w.r.t. the parameter $n$ for deciding $k$-ESU, and a more efficient (compared to the similar algorithm presented in \cite{regunivpaper}) FPT algorithm w.r.t. parameter $\sigma$. In both cases, we assume that $L$ does not contain words of arbitrarily large universality, so we can also assume that $k\leq n$, as explained above. 

\medskip

\noindent \textbf{FPT Algorithm w.r.t. $n$.}
To begin with, we show the following theorem:
\begin{theorem}\label{thm:reducedwalkExpressions}
If $L$ contains a word $w$ with $\iota(w)=k$ then there exists a word $u\in L$ with $\iota(u)\geq k$, which is the label of a walk denoted by a walk expression of the form $(q_0[c_{q_0}]q_0',b_1,q_1)(q_1[c_{q_1}]q_1',b_2,q_2)\cdots (q_{h-1}[c_{q_{h-1}}]q'_{h-1},b_h,q_h)$, where: $h\leq n$; $q'_i$ belongs to ${\mathcal C}_{S(q_i)}$ for all $i\in [h-1]$; $S(q_i)\neq S(q_j)$, for all $i\neq j$; and $q_h=f$ is the final state of $A$. 
\end{theorem}
\input{proofs/proof_reducedPathExpressions}

Note that we will not use Theorem \ref{thm:reducedwalkExpressions} as an algorithmic tool. It just shows that to identify the maximum universality index over the words of $L$ it is enough to consider words $w$ of maximum universality which label walks denoted by walk expressions $(q_0[c_{q_0}]q_0',b_1,q_1)\cdots (q_{h-1}[c_{q_{h-1}}]q'_{h-1},b_h,q_h)$, where: $h\leq n$; $q'_i$ belongs to $S(q_i)$ for all $i\in [h-1]$; $S(q_i)\neq S(q_j)$, for all $i\neq j$; and $q_h=f$ is the final state of $A$. In the following, we show that the search space for the words $w\in L$ with maximum universality index can be restricted further. In particular, our goal is to show that, in order to compute the maximum universality of a word accepted by $A$ along a walk described by expressions as the ones above, it is not important how many times a cycle is traversed, but it is only important to know the states $q$ along this walk and their corresponding sets $V_{S(q)}$. 

Before stating our next result, note that the walks denoted by an expression $\beta'=(q_0[c_{q_0}]q_0',b_1,q_1)(q_1[c_{q_1}]q_1',b_2,q_2)\cdots (q_{h-1}[c_{q_{h-1}}]q'_{h-1},b_h,q_h)$ are exactly the walks $\beta=c^{i_0}_{q_0}p_{q_0,q'_0}(q'_0,,b_1,q_1)\cdots c^{i_{h-1}}_{q_{h-1}}p_{q_{h-1}, q'_{h-1}}(q_{h-1},b_h,q_h)$, where $i_0,\ldots,i_{h-1}$ are non-negative integers. Our next result analyses certain factors of such walks.

\begin{lemma}\label{lem:decomposition}
Let $w$ be a word with $\iota(w)=k$, which labels a walk \\
\centerline{$\beta=c^{i_1}_{q_1}p_{q_1,q'_1}(q'_1,b_2,q_2)c^{i_2}_{q_2}p_{q_2,q'_2}(q'_2,b_3,q_3)\cdots c^{i_{t-1}}_{q_{t-1}}p_{q_{t-1},q'_{t-1}}(q'_{t-1},b_t,q_t)c^{i_{t}}_{q_t}$,} 
where $q'_i$ is a state of ${\mathcal C}_{S(q_i)}$ for all $i\in [t]$, and $S(q_i)\neq S(q_j)$ for $i\neq j$. 

Then, there exists a word $w'=w'_1\cdots w'_k$ and integers $s_0,s_1,\ldots,s_k$ such that:
\begin{itemize}[leftmargin=10pt,itemsep=1cm,parsep=-1cm]
\item $1=s_0<s_1<\ldots<s_k=t$
\item For $i\in [k]$, $w'_i$ is $1$-universal and labels the walk:\\
\hspace*{-10pt} $c_{q_{s_{i-1}}} p_{q_{s_{i-1}},q'_{s_{i-1}}}(q'_{s_{i-1}},b_{s_{i-1}+1},q_{s_{i-1}+1})\cdots c_{q_{s_i-1}}  p_{q_{s_i-1},q'_{s_i-1}} (q'_{s_i-1},b_{s_i},q_{s_i})c_{q_{s_i}}.$
		\vspace*{-5pt}
\end{itemize}
\end{lemma}
\input{proofs/proof_decomposition}

Therefore, by Lemma \ref{lem:decomposition}, to compute the maximum universality index over the words of $L$ it is enough to identify a word $w'\in L$ of maximum universality that labels some walk obtained by concatenating several subwalks of the form $c_{q_1}p_{q_1,q'_1} (q'_1,b_2,q_2)\cdots c_{q_{h-1}}p_{q_{h-1},q'_{h-1}}(q'_{h-1},b_h,q_h)c_{q_h}$, where: $h\leq n$; $q'_i$ belongs to $S(q_i)$ for all $i\in [h-1]$; $S(q_i)\neq S(q_j)$, for all $i\neq j$; and $q_h=f$ is the final state of $A$.
The key observation is that, in fact, to find such a word of maximum universality, it is enough to know the sequence of states $q_1,\ldots,q_h$ as above. We will explain now how this is done efficiently. 

\begin{lemma}\label{lem:matching}
Given a sequence of $h<n$ states $q_1\cdots q_h$, such that $S(q_i)\neq S(q_j)$ for all $i\neq j$, we can decide in $O(m+n\sigma \sqrt{n})$ time whether there exist states $g_1,\ldots, g_h$ and $g'_1,\ldots, g'_h$, with $g_1=q_1$, $S(g_i)=S(g'_i)=S(q_i)$ for all $i\in [h]$, and the walk $\beta = c_{q_1}p_{q_1,g'_1} (g'_1,b_2,g_2)\ldots c_{g_{h-1}}p_{g_{h-1},g'_{h-1}}(g'_{h-1},b_h,g_h)c_{g_h}p_{g_h,q_h}$ such that $\labels(\beta)$ is $1$-universal. 
\end{lemma}
\input{proofs/proof_matching}

Note that in Lemma \ref{lem:matching} we do not compute the sequence $g_1\cdots g_h$, we just decide whether there is a walk from $q_1$ to $q_h$, going through the strongly connected components of $q_2,\ldots,q_{h-1}$, which is labelled by a $1$-universal word. 

Now, we can state a result which is the main building block of our algorithm.
\begin{lemma}\label{lem:algo_one_seq}
Given a sequence $q_1\cdots q_h$ of $h< n$ states with $q_h=f$ and $S(q_i)\neq S(q_j)$, for all $i\neq j$, we can compute in $O(n(m+ n\sigma \sqrt{n}))$ time the maximum integer $k$ for which there exists a walk $\beta$ with origin $q_1$ and target $q_h$ such that $\labels(\beta)$ is $k$-universal and $\beta$ goes through the strongly connected components of $q_1,\ldots,q_h$, and through no other strongly connected component. 
\end{lemma}
\ifpaper
\else
\input{proofs/proof_algo_one_seq}
\fi

We can now state the main result of this section.
\begin{theorem}\label{thm:FPT_states}
Given a regular language $L$, over $\Sigma$, with $|\Sigma|=\sigma$, specified as an NFA $A$ with $n$ states, we can solve $k$-ESU in $O(2^nn(m+ n\sigma\sqrt{n}))$ time. That is, in the case when the input is given as an NFA, $k$-ESU is FPT w.r.t. the number $n$ of states of the input NFA.
\end{theorem}
\ifpaper
\else
\input{proofs/proof_FPT_states}
\fi
The algorithm of this theorem considers each set of states $\{q_0,f\}\cup \{q_1,\ldots,q_{h-1}\}$ of $A$, sorts it w.r.t. the topological sorting of $S(A)$ to obtain the sequence $q_0q_1\ldots q_{h-1}f$, and then uses Lemma \ref{lem:algo_one_seq} on it, if the prerequisites of that lemma apply. We return the largest universality index computed for such a subsequence. \looseness=-1

Note that Theorem \ref{thm:FPT_states} shows that $k$-ESU can be solved by an FPT-algorithm w.r.t. the number of states of the input NFA. Moreover, as $m\in O( n^2\sigma)$, $k$-ESU can be solved in linear time $O(\sigma)$ for NFAs with $n\in O(1)$. 

%% file: proofs/proof_infUniv.tex
\begin{proof}
	If such a state $q$ exists, then we get immediately the result, given that $q$ is both accessible and co-accessible. To get a word $w_i$ with $\iota(w_i)\geq i$, we follow a path from $q_0$ to $q$, then follow the cycle $c_q$ for $i$ times, and then follow a path from $q$ to the final state of $A$. The converse follows by a pumping argument. Recall that the automaton $A$ has $n$ states and we have in $L$ a word of universality $\iota(w_{n+1})= n+1$. Then $w_n=u_1\cdots u_{n+1}$, where $\al(u_i)=\Sigma$ for all $i\in [n+1]$. Moreover, the word $w_{n+1}$ is accepted along a path $\beta_1\beta_2\cdots \beta_{n+1}$, where $\labels(\beta_i)=u_i$ for all $i\in [n+1]$. By the pigeonhole principle, there exist $i$ and $j$, with $i,j\in [n+1]$ and $i<j$, such that the origin of $\beta_i$ and $\beta_j$ is the same state $q$. Then, $c=\beta_i\cdots \beta_{j-1}$ is a cycle that goes through $q$ and contains only transitions of ${\mathcal C}_{S(q)}$, so this means that $\al(\labels(c_q))=\Sigma$. \qed
\end{proof}

%% file: proofs/proof_cycles.tex
\begin{proof}
We execute first the preprocessing steps from Section \ref{sec:preProc}. Then, for each strongly component ${\mathcal C}_i$ of $A$, we initialize $V_i=\emptyset$. We then go through the set of transitions of $A$, and for every transition $(q_1,a,q_2)$ with $S(q_1)=S(q_2)$, we add $a$ to $V_{{S(q_1)}}$ (note that these sets might contain $\epsilon$). This can be done in $O(m)$ time. Finally, for some state $q\in Q$, we can access $V_q$ by accessing $V_{S(q)}$ and, if needed, removing $\epsilon$ from $V_{S(q)}$.
The conclusion follows.	
	\qed 
\end{proof}

%% file: proofs/proof_reducedPathExpressions.tex
\begin{proof}
	Let us consider a walk $\beta=(q_0,a_1,q_1)(q_1,a_2,q_2)\cdots (q_{r-1},a_r,q_r)$ with label $w$, where $q_r=f$. We will rewrite this walk in several steps, in order to derive the result from the statement.
	
	Clearly, the walk expression \\
	\centerline{$\beta'=(q_0[c_{q_0}]q_0,a_1,q_1)(q_1[c_{q_1}]q_1,a_2,q_2)\cdots (q_{r-1}[c_{q_{r-1}}]q_{r-1},a_r,q_r)$} 
	denotes walks whose labels are words $w'$ for which $\iota(w')\geq \iota(w)$ (as all such words have $w$ as subsequence). Obviously, $w$ is the label of some walk (namely $\beta$) denoted by the walk expression $\beta'$.
	
	We iterate the following process on $\beta'$:
\begin{itemize}[leftmargin=10pt,itemsep=1cm,parsep=-1cm]
		\item Firstly, find the leftmost atom $(q_1[c_{q_1}]q_1,a,q'_1)$ of the walk expression $\beta'$ such that there exists another atom $(q_2[c_{q_2}]q_2,b,q'_2)$ in $\beta'$, with $S(q_1)=S(q_2)$. Note that, if such states $q_1,q_2$ exist, then $q_1\neq f\neq q_2$. 
		\item Secondly, find the rightmost atom  $(q_2[c_{q_2}]q_2,b,q'_2)$, with $S(q_2)=S(q_1)$. Thus: $\beta' = x (q_1[c_{q_1}]q_1,a,q'_1) z (q_2[c_{q_2}]q_2,b,q'_2) y$, where neither $x$ nor $y$ contain atoms originating with a state $q$ of ${\mathcal C}_{S(q_1)}$ and $z$ is a walk visiting only states in ${\mathcal C}_{S(q_1)}$. Rewrite $\beta'$ as $\beta' = x (q_1[c_{q_1}]q_2,b,q'_2) y$. Note that no state $q$ of ${\mathcal C}_{S(q_1)}$ may appear in $\beta'$ as the origin of a walk expression, except $q_1$. Moreover, if a state $g_1$ is the origin of some atom of $x  (q_1[c_{q_1}]q_2,b,q'_2) $, then there is no other atom of $\beta'$ whose origin is some state $g'$ of ${\mathcal C}_{S(g)}$. 
		\vspace*{-5pt}
	\end{itemize}
	
	Let $w_0=w$. Now, we claim that at the end of the $i^\text{th}$ iteration, for $i\geq 1$, there exists a word $w_i$ which labels one of the walks denoted by the current expression $\beta'$ (i.e., as computed at the end of the respective iteration) such that $\iota(w_i)\geq \iota(w_{i-1})$. Indeed, we can obtain $w_i$ as follows: we take the prefix of $w_{i-1}$ corresponding to the walk expression $x$, then we follow transitions of the cycle $c_{q_1}$ at least $2|z|+1$ times (which is more than enough to cover all the factors of $z$ that had some state $q$ of ${\mathcal C}_{S(q_1)}$ as origin), then we follow the transition $(q_2,b,q'_2)$, and end with the suffix of $w_{i-1}$ corresponding to the expression $y$. 
	
	The above process is finite (in each step we reduce the number of atoms of $\beta'$), and shows how each walk $\beta$, with label $w$, is transformed into a walk expression $\beta'=(q_0[c_{q_0}]q_0',b_1,q_1)\cdots (q_{h-1}[c_{q_{h-1}}]q'_{h-1},b_h,q_h)$, where: $h\leq n$; $q'_i$ belongs to ${\mathcal C}_{S(q_i)}$ for all $i\in [h-1]$; $S(q_i)\neq S(q_j)$, for all $i\neq j$; and $q_h=f$. Moreover, as shown, there exists at least one word $u$ labelling a walk denoted by the expression $\beta'$, derived at the end of this process, such that $\iota(u)\geq \iota(w)$. 
	\qed\end{proof}

%% file: proofs/proof_decomposition.tex
\begin{proof}
	As $\iota(w)\geq k$, we have that there exist words $w_1,\ldots, w_k$, with $w=w_1\cdots w_k$ and $\iota(w_i)=1$ for all $i\in [k]$. Let $\beta_i$ be the subwalk of $\beta$ labelled by $w_i$. We have $\beta=\beta_1\beta_2\cdots \beta_k$. For simplicity, let $q'_t=q_t$, $p_{q_t,q'_t}=\epsilon$, and $s_0=1$. \looseness=-1
	
	For $i\in [k]$, let $s_i\geq s_{i-1}$ be the smallest integer such that $\beta_i$ is a factor (that is, contiguous subwalk) of $c^{i_{s_{i-1}}}_{q_{s_{i-1}}} p_{q_{s_{i-1}},q'_{s_{i-1}}}(q'_{s_{i-1}},b_{s_{i-1}+1},q_{s_{i-1}+1})\cdots c^{i_{s_i}}_{q_{s_i}} p_{q_{s_i},q'_{s_i}}$. It is not hard now to see that $s_{i}>s_{i-1}$ holds, as $\al(\labels(c_{q_{s_{i-1}}} p_{q_{s_{i-1}},q'_{s_{i-1}}})) = \al(\labels(c_{q_{s_{i-1}}} ))\subsetneq \Sigma = \al(\labels(\beta_i))= \al(w_i)$. 
	
	Let us now define, for all $i\in [t]$, the walk $\beta'_i$ as:\\
\centerline{$ c_{q_{s_{i-1}}} p_{q_{s_{i-1}},q'_{s_{i-1}}}(q'_{s_{i-1}},b_{s_{i-1}+1},q_{s_{i-1}+1})\cdots c_{q_{s_i-1}}  p_{q_{s_i-1},q'_{s_i-1}} (q'_{s_i-1},b_{s_i},q_{s_i})c_{q_{s_i}}$.}
	Also, let $w'_i=\labels(\beta'_i)$. Clearly, $\al(w'_i)=\al(w_i)=\Sigma$, as $\al(\label(c^{i_{s_i}}_{q_{s_i}} p_{q_{s_i},q'_{s_i}})=\al(\label(c_{q_{s_i}}p_{q_{s_i},q'_{s_i}})=\al(\label(c_{q_{s_i}})$. The conclusion follows, for $w'=w'_1\cdots w'_k$. 
	\qed 
\end{proof}

%% file: proofs/proof_matching.tex
\begin{proof}
	Firstly, we can compute in $O(n+m)$ time the sets $V_{q_{i},q_{i+1}}\subseteq \Sigma\cup\{\epsilon\}$ of labels of transitions from states of ${\mathcal C}_{S(q_i)}$ to states of ${\mathcal C}_{S(q_{i+1})}$, for all $i\in [h-1]$. 
	If there exists $i\in [h-1]$ such that $V_{q_{i},q_{i+1}}=\emptyset$, then we answer the considered problem negatively.  Let us assume, in the following, that this is not the case.
	
	Then, we compute the sets $V_{S(q_i)}$, for all $i\in [h]$, as described in Lemma \ref{lem:cycles}, and the set $V=\bigcup_{i\in h} V_{S(q_i)}$. This takes $O(n+m)$ time. Further, we want to see if there is a way to identify states $g_2,\ldots, g_h,$ $g'_1,\ldots, g'_h$, and the letters $b_1,\dots, b_h$, such that $\beta = c_{q_1}p_{q_1,g'_1} (g'_1,b_2,g_2)\ldots c_{g_{h-1}}p_{g_{h-1},g'_{h-1}}(g'_{h-1},b_h,g_h)c_{g_h}p_{g_h,q_h}$ is $1$-universal. In other words, $ \Sigma\setminus V\subseteq \{b_1,\dots, b_h\}$. 
	
	The first case is when $V=\Sigma$. In that particular case, all walks $\beta = c_{q_1}p_{q_1,g'_1} (g'_1,b_2,g_2)\cdots c_{g_{h-1}}p_{g_{h-1},g'_{h-1}}(g'_{h-1},b_h,g_h)c_{g_h}p_{g_h,q_h}$, fulfilling the conditions from the statement, are labelled by $1$-universal words, and the considered problem can be answered positively. 
	
	The second case is if $\Sigma\subset V$ and $|V| >h$; then, the answer is negative.
	
	Otherwise, we define a bipartite graph $G$ as follows:
\begin{itemize}[leftmargin=10pt,itemsep=1cm,parsep=-1cm]
		\item The set of vertices is defined by the following two disjoint sets: on the one side, we have the set $[h-1]$ (so the vertices are integers $i\in [h-1]$)  and, on the other side, the set $\Sigma \setminus V$ (so the vertices are letters $b\in \Sigma \setminus V$).
		\item There exist an edge between vertex $i\in [h-1]$ and vertex $b\in \Sigma \setminus V$ if and only if $b\in V_{q_i,q_{i+1}}$. $G$ contains no other edges.
				\vspace*{-5pt}
	\end{itemize}
	
	If $\nu=|\Sigma \setminus V|$, we have $\nu \leq h$, so we can compute in $O(h\nu \sqrt{h+\nu})\subseteq O(n\sigma\sqrt{n})$ time a maximum-cardinality bipartite matching \cite{HopcroftK73}. If the cardinality of this matching is $\nu$ (that is, every letter-vertex of $G$ is covered by one edge of the matching), then we proceed as follows. 
	We first define, for $i\in \{2,\ldots,h\}$, the letter $b_i$: if $(i-1,b)$ is an edge of the matching, then $b_i=b$; otherwise, $b_i$ is an arbitrary element of $ V_{q_{i-1},q_i}$. In both cases, we set $g'_{i-1}\in V_{S(q_{i-1})}$ and $g_i\in V_{S(q_{i})}$ to be states connected by the letter $b_i$. It is now immediate that the path  $\beta = c_{q_1}p_{q_1,g'_1} (g'_1,b_2,g_2)\cdots c_{g_{h-1}}p_{g_{h-1},g'_{h-1}}(g'_{h-1},b_h,g_h)c_{g_h}p_{g_h,q_h}$ is labelled by an $1$-universal word: each letter of $\Sigma$ is either the label of a transition from a cycle $c_{q_i}$ or of a transition between some states $g'_{i-1}$ and $g_i$. Thus, in this case, we can answer the problem positively. If the cardinality of this matching is strictly smaller than $\nu $, then the problem is answered negatively. In that case, we cannot cover all letters with either transitions of the strongly connected components ${\mathcal C}_{S(q_i)}$, for $i\in [t]$, or the transition between states of these components. 
	
	Therefore, the statement holds.
	\qed \end{proof}

%% file: proofs/proof_algo_one_seq.tex
\begin{proof}
	The problem can be solved by a dynamic programming algorithm. We initialise an integer-array $D$ with $n$ elements, where $D[i]=0$ for all $i$.
	
	Firstly, for every $i\in [h]$, we consider $j$ from $1$ to $i$.  We check with Lemma \ref{lem:matching}, for the sequence $q_j\cdots q_i$ whether there exists a path $\beta_{j,i}$ which is $1$-universal and connects $q_i$ to $q_j$, and goes through the strongly connected components of $q_i,\ldots,q_j$, but through no other strongly connected component (as in the statement of the respective lemma). If such a path $\beta_{j,i}$ exists, we set $D[i]=\max\{D[i],D[j]+1\}$. 
	
	It is not hard to show (by induction) that at the end of the iteration for $i=t$, we have that, for all $j\leq t$, the value $D[j]$ equals the maximum number $k_j$ for which there exists a path $\beta_j$, which connects $q_1$ to $q_j$, while going through the strongly connected components of $q_1,\ldots,q_j$, and through no other strongly connected component, such that $\labels(\beta_j)$ is $k_j$-universal. So, at the end of the algorithm, once the iteration for $i=h$ is completed, $D[h]$ equals the maximum number $k$ for which there exists a path $\beta$ connecting $q_1$ to $q_h$, going through the strongly connected components of $q_1,\ldots, q_h$ but through no other strongly connected component, such that $\labels(\beta)$ is $k$-universal (recall that $c_{q_h}$ is empty). 
	
	The running time of this algorithm is $O(n^2(m+ n\sigma \sqrt{n}))$, as it simply executes $O(h^2)$ times the algorithm from Lemma \ref{lem:matching} and $h\leq n$. 
	
	This can be further optimised as follows. To decide whether there exists a path $\beta_{i}$, going through the strongly connected components of $q_1,\ldots,q_i$, whose label is $k$-universal, our algorithm simply identifies the smallest $j$ such that there exist a path $\beta_{j}$, going through the strongly connected components of $q_1,\ldots,q_j$, whose label is $(k-1)$-universal, and then checks the existence of $\beta_{i,j}$ (as defined above) which is $1$-universal. Note that if such a $\beta_{i,j}$ does not exist for the minimal $j$ for which a path $\beta_{j}$ whose label is $(k-1)$-universal, then it will also not exist for any greater $j$. Therefore, a greedy strategy can be used to solve our problem. Using the notations introduced above, we first find the smallest $i_0$ such that there exist a path $\beta_{i_0}$ whose label is $1$-universal. Then, we find the smallest $i_i>i_0$ such that there exist a path $\beta_{i_0,i_1}$ whose label is $1$-universal, and conclude that $i_1$ is the smallest integer such that there exists $\beta_{i_1}$ (obtained by the concatenation of $\beta_{i_0}$ and $\beta_{i_0,i_1}$) whose label is $2$-universal; if there is no position $i$ for which a path $\beta_{i_0,i}$ with $1$-universal label exists, we set $i_1=-1$. We repeat this process, and, for $2\leq j\geq n$, if $i_{j-1}\neq -1$, identify the smallest $i_j$ such that there exists $\beta_{i_j}$ (obtained by the concatenation of $\beta_{i_{j-1}}$ and $\beta_{i_{j-1},i_j}$) whose label is $j$-universal; again, if we do not find any $i_j\leq n$ with the desired property, we set $i_j=-1$. We return the largest $j$ for which $i_j\neq -1$. This process only takes $O(n(m+ n\sigma\sqrt{n}))$ time, as it simply executes at most $O(n)$ times the algorithm from Lemma \ref{lem:matching}.
	Thus, the statement clearly holds.
	\qed\end{proof}

%% file: proofs/proof_FPT_states.tex
\begin{proof}
	As discussed at the beginning of this section, we can assume that $A$ has a single final state. We process $A$ as in Lemma \ref{lem:cycles} and obtain the strongly connected components of $A$, and the corresponding directed acyclic graph of strongly connected components $S(A)$. We now order the vertices of $S(A)$ topologically (using the standard topological sorting algorithm, which runs in $O(n+m)$ time). Note that ${\mathcal C}_{S(q_0)}$ is always the first strongly connected component in this topological sorting of $S(A)$ (as all states are accessible) and ${\mathcal C}_{S_f}=\{f\}$ is the last component in this sorting. 
		
	We now iterate over all subsets $Q'=\{q_0,f\}\cup \{q_1,\ldots,q_{h-1}\}$ of $Q$; for simplicity, set $q_h=f$ and note that $h< n$. 
	
	For such a set, we check whether $S(q_i)\neq S(q_j)$, for all $i\neq j$; if this is not the case, we consider the next subset of $Q$. This can be done in $O(e)\subseteq O(n)$ time, by marking in an array of length $e$ how many times each integer $i\in [e]$ (that is, label of a strongly connected component) appears in the set $Q'=\{S(q_0),S(f)\}\cup \{S(q_1),\ldots,S(q_{h-1})\}$.
	
	If no two states of the set $Q'$ are part of the same strongly connected component, we sort $Q'$ w.r.t. the topological sorting of $S(A)$ (which can be done in $O(n)$ time, as we already have the topological sorting of $S(A)$). 
	
	After this step, and a potential relabelling of the states, we can assume that $q_{i-1}$ comes before $q_{i}$ in the topological sorting, for all $i\in [h]$, $q_0$ is the first in this order, and $q_h=f$ still holds. Then, we simply use Lemma \ref{lem:algo_one_seq} for the sequence $q_0\cdots q_h$ of $h< n$ states, as the conditions from the statement of this lemma are fulfilled. We then store the maximum value $k$ computed in the respective lemma (over all the considered sequences $q_0\cdots q_h$). This is the largest universality of a word which labels a path connecting $q_0$ to $f$, so the maximum universality of a word of $L$. As the number of the considered subsets is $O(2^n)$, the statement follows based on the complexity of the algorithm of Lemma~\ref{lem:algo_one_seq}.
	\qed \end{proof}

%% file: acyclic.tex
\medskip

\noindent \textbf{FPT Algorithm w.r.t. $\sigma$.} In \cite{regunivpaper} it was shown that we can solve $k$-ESU in $O(n^3 \poly(\sigma) 2^\sigma)$ time. We improve this result and give a linear time solution for $k$-ESU in the case of constant alphabets.
\begin{theorem}\label{thm:FPT_sigma}
Given a regular language $L$, over $\Sigma$, with $|\Sigma|=\sigma$, specified as an NFA $A$ with $n$ states, we can solve $k$-ESU in $O(n + m + m2^\sigma)$ time. 
\end{theorem}
\ifpaper
\else
\input{proofs/proof_FPT_sigma}
\fi

%% file: proofs/proof_FPT_sigma.tex
\begin{proof}
	We once more assume that $A$ has a single final state. We process $A$ as in Section \ref{sec:preProc} and obtain the strongly connected components of $G_A$, and the corresponding directed acyclic graph of strongly connected components $S(A)$. We now order the vertices of $S(A)$ topologically (using the standard topological sorting algorithm, which runs in $O(n+m)$ time). Note that ${\mathcal C}_{S(q_0)}$ is always the first strongly connected component in this topological sorting of $S(A)$ (as all states are accessible) and ${\mathcal C}_{S_f}=\{f\}$ is the last component in this sorting. 
	
	Let ${\mathcal C}_1={\mathcal C}_{S(q_0)},\ldots, {\mathcal C}_e=\{f\}$ be the strongly connected components of $A$, as ordered by the topological sorting. It is worth noting that in the directed acyclic graph of the strongly connected components of $A$, ${\mathcal C}_e$ is the only component that it is not the origin of any edge.
	
	Using Lemma \ref{lem:cycles}, we compute $V_i$, the set of labels of transitions of ${\mathcal C}_i$, for $i\in [e]$. For simplicity, if some $V_i$ contains $\epsilon$, then we remove it. Note that $\Sigma\setminus V_i\neq \emptyset$, as, otherwise, we would have words of arbitrarily large universality in $L$, and we have assumed that this is not the case. Similarly, for $j\in [e]$, we compute the set $V_{\rightarrow j}$ of pairs $(i,a)$ for which there exist transitions $(q_1,a, q_2)$ with $q_1\in {\mathcal C}_i$ and $q_2\in {\mathcal C}_j$. This can be done in $O(n+m)$ time as follows. We go through the transitions of $A$ and rename the transition $(q_1,a,q_2)$ as $(S(q_1),a,S(q_2)$. We sort (using radix sort) the list of renamed transitions $(S(q_1),a,S(q_2))$ of $A$, w.r.t. their third component (i.e., target), as first criterium, their first component (i.e., origin) as second criterium, and the second component (i.e., label) as third criterium (for this criterium we take the natural order on the integer alphabet $\Sigma$, and assume $\epsilon$ is smaller than all letters of $\Sigma$). Then, we eliminate duplicates. Finally, we go through the sorted list and if $(i,a,j)$ is the current element, then we add $(i,a)$ to $V_{\rightarrow j}$ 
	
	We now implement a dynamic programming approach. We define $e\times 2^\sigma$ matrix $D[\cdot][\cdot]$, where $D[i][S]$, for some $i\in [e]$ and $V_i \subseteq S\subsetneq \Sigma$, is the largest universality index of a word $w$ that labels a path which goes through the components ${\mathcal C}_1, \ldots, {\mathcal C}_i$, ending in ${\mathcal C}_i$, such that the rest $\al(\r(w))=S$. 
	
	We initialize $D[i][S]=-1$, for all $i\in [e]$ and $S\subsetneq \Sigma$.
	
	In step $1$, we set $D[1][V_1]=0$. 
	
	In step $i$, for $i\geq 2$, we compute the values $D[i][\cdot]$ by the following formula. If $(h,a)\in V_{\rightarrow j}$, we have $h<i$ due to the topological sorting; let $I_a=\{a\}$ if $a\neq \epsilon$ and $I_a=\emptyset$, otherwise. For each $S'\subsetneq \Sigma$, we proceed as follows. If $S'\cup\{a\}\cup V_i=\Sigma$ and $D[h][S']\neq -1$, then we set $D[i][V_i]=D[h][S']+1$. If $S'\cup\{a\}\cup V_i\neq \Sigma$ and $D[h][S']\neq -1$, we set $D[i][S'\cup\{a\}\cup V_i]=D[h][S']$. 
	
	The time needed in each step $i$ is upper bounded by $O(m_i 2^\sigma)$, where $m_i$ is the number of transitions whose target is in ${\mathcal C}_i$, so the overall time of computing the elements of the matrix $D$ is $O(m 2^\sigma)$.
	
	The largest value stored in the array $D[e][\cdot]$ is the largest universality of a word of $L$. To solve $k$-ESU, we simply need to check whether this value is greater than $k$. \looseness=-1
	
	The correctness of our approach is based on the following argument. We are interested in computing the words with the largest universality index that label paths which go through the components ${\mathcal C}_1, \ldots, {\mathcal C}_i$, and end in ${\mathcal C}_i$. For $i=1$, any such word $w$ has the rest $\r(w)\subseteq V_1$ and is not $1$-universal; obviously, it makes sense to always consider the words which have the largest rest w.r.t. inclusion. That is, we will consider only words which have $\r(w)=S$ (e.g.., which go at least once through all the transitions of ${\mathcal C}_1$); therefore we set $D[1][V_1]=0$. Clearly, $D[1][S]$ is undefined for all $S$ strictly containing $V_1$. Now, in step $i\leq 2$, let us assume the values $D[j][\cdot]$ are correctly computed for all $j<i$, and we want to compute the words with the largest universality index that label paths which go through the components ${\mathcal C}_1, \ldots, {\mathcal C}_i$, and end in ${\mathcal C}_i$. Let $w$ be such a word. We have that $w=w'aw''$, were $w'$ labels a path which goes through the components ${\mathcal C}_1, \ldots, {\mathcal C}_j$, for some $j$, $a$ is the label of a transition connecting component ${\mathcal C}_j$ with component ${\mathcal C}_i$, and then $w''$ is the label of a path going through $C_i$. If $\r(w')=S$, then we can assume that $w'$ is one of the words with the largest universality index that label paths which go through the components ${\mathcal C}_1, \ldots, {\mathcal C}_j$, ends in ${\mathcal C}_j$, and $\r(w')=S$ (i.e., its universality index is $D[j][S]$). Now, following the transition labelled with $a$ and those labelled with $w''$ we can achieve two things. On the one hand, if $V_i\cup\{b\}\cup S=\Sigma$, we can complete $\r(w' a)$ with the letters from $V_i$ (by choosing a prefix of $w''$ to be the label of a path that goes through all the transitions of ${\mathcal C}_i$) and create a new arch; then we can assume that $w''$ continues with the label of a path that goes through all the transitions of ${\mathcal C}_i$, and this creates a rest whose alphabet is $V_i$ (we could also assume $w''$ ends with a path with less letters, but this does not make sense when creating words with large universality index). To conclude, this suffix of $w''$ cannot create more arches nor a rest with a larger alphabet (w.r.t. inclusion); however, in total, $w''$ created one additional arch and the with alphabet equal to  $V_i$. On the other hand, if $V_i\cup\{b\}\cup S\neq \Sigma$, $w''$ can only contribute by adding the letters of $V_i$ to the alphabet of the rest of the constructed word (by choosing $w''$ to be the label of a path that goes through all the transitions of ${\mathcal C}_i$). These two cases are covered by the formula implemented by our dynamic programming approach. Clearly, when $i$ is considered we have to consider all choices of the component ${\mathcal C}_j$ from which we make the transition to ${\mathcal C}_i$. Based on these remarks, our approach is sound. 
	
	The statement follows.
	\qed\end{proof}

%% file: regex.tex
\medskip

\noindent \textbf{Regular Expressions.} 
We now consider languages represented by regexes. While NFAs and regexes define the same class of languages, the size of a regex required to represent a given language can be exponentially larger than the size of the NFA accepting that language (and vice versa, see, e.g., \cite{GruberH15}). 
In this section, we first show that $k$-ESU, where the language is specified as a regex of length $n$, can be solved in $O( n 2^{\sigma})$ time, using the FPT-algorithm w.r.t. $\sigma$ described above. Then, we provide a new proof of the NP-completeness of $k$-ESU when the input language is given as a regex. This result is interesting, as the hardness proof from \cite{regunivpaper} seems to require, in some cases, a regex exponentially larger than the corresponding NFA to describe the regular language used in the given reduction.  \looseness=-1

We start with the algorithmic part of this section. We assume from now on that we are given a regular expression of length $n$. The first result is immediate, by a proof similar to Lemma \ref{lem:infUniv}. 
\begin{lemma}
    \label{lem:unbounded_universality_regex}
    For regular expression $R$, if $R=R_1(R_2)^*R_3$ and $\Sigma\subseteq \al(R_2)$, then $L(R)$ is $k$-universal, for every $k\geq 1$. 
\end{lemma}
We can now decide in $O(n)$ time, for a regex $R$ of length $n$, whether $L(R)$ is $k$-universal, for every $k\geq 1$. If not, then $\iota(w)\leq n$ for all $w\in L(R)$, as $|w|\leq n$. In this case, we can reduce the case of general regexes to the case of star-free regexes.\looseness=-1

\begin{lemma}\label{lem:regex_acyclic}
    Given regular expression $R$, if there does not exist a decomposition $R = U_1(U_2)^*U_3$ with $\Sigma\subseteq \al(U_2)$, then there exist $k\geq 1$ such that $L(R)$ is not $k$-universal. Moreover, there exists a star-free regular expression $R'$, which can be computed in linear time $O(|R|)$, such that $\max\{\iota(w)\mid w\in L(R')\}= \max\{\iota(w)\mid w\in L(R)\}$.
\end{lemma}
\ifpaper
\else
\input{proofs/proof_regex_acyclic}
\fi

Lemma \ref{lem:regex_acyclic} immediately leads us to the following conclusion. 
\begin{theorem}\label{thm:FPT_sigma_regex}
Given a regular language $L$, over $\Sigma$, with $|\Sigma|=\sigma$, specified as a regex $R$ of length $n$, we can solve $k$-ESU in $O(n2^\sigma)$ time. 
\end{theorem}
\ifpaper
\else
\input{proofs/proof_FPT_sigma_regex}
\fi

This result shows that $k$-ESU can be solved in linear time for regular expressions over constant-size alphabets.

Regarding the hardness of $k$-ESU when the input is given as a regular expression, we can show the following result. 
\begin{theorem}
    \label{thm:regex_is_hard}
    $k$-ESU is NP-complete for languages defined by regular expressions. The problem is already NP-hard for $k=1$ and star-free regexes. 
\end{theorem}
\ifpaper
\else
\input{proofs/proof_regex_is_hard}
\fi

\medskip 

\noindent \textbf{Concluding Remarks.} The lower bound of Theorem \ref{thm:regex_is_hard} and its proof allow us to also make several final remarks on the presented FPT algorithms for $k$-ESU. We consider first the case when the input regular language is given as an NFA with $n$ states over an alphabet with $\sigma$ letters. The lower bound derived for $k$-ESU in~\cite{regunivpaper} and its proof show that, under the Exponential Time Hypothesis \cite{ImpagliazzoPZ01,LokshtanovMS11}, there are no $2^{o(\sigma)}\poly(n,\sigma)$ time algorithms solving the respective problem. The reduction from the proof of Theorem \ref{thm:regex_is_hard}, together with the Sparsification Lemma~\cite{ImpagliazzoPZ01}, shows that there are no $2^{o(n)}\poly(n,\sigma)$ time algorithms for $k$-ESU, under ETH. If the input language is given as a regular expression of length $n$ over an alphabet with $\sigma$ letters, we can once more refer to the reduction from the proof of Theorem \ref{thm:regex_is_hard} and the Sparsification Lemma, and get that there are no $2^{o(n)}\poly(n,\sigma)$ time algorithms for $k$-ESU, in this setting. As such, the algorithmic results we have obtained are tight from a fine-grained complexity perspective. \looseness=-1

%% file: proofs/proof_regex_acyclic.tex
\begin{proof}
	Let $R_0=R$. For $i\geq 0$ we rewrite $R_i$ into $R_{i+1}$ as follows. If $R_i=U_1(U_2)^*U_3$, such that $U_1$ is a valid star-free regex (i.e., there is open parenthesis which is not matched inside $U_1$), then we define $R_{i+1}=U_1u_2u_2U_3$, where $u_2=a_1\cdots a_t$ and $\{a_1,\ldots,a_t\}=\al(U_2)\cap\Sigma$. We claim that $\max\{\iota(w)\mid w\in L(R_{i+1})\}= \max\{\iota(w)\mid w\in L(R_{i})\}$. 
	
	Indeed, let $w$ be a word from $L(R_i)$ such that $\iota(w)$ is maximum. Then $w=w' w_1\cdots w_p w''$, such that $w'$ is described by $U_1$, $w''$ by $U_3$, and $w_1,\ldots,w_p$ are all described by $U_2$. Consider the arch factorisation of $w$: $w=\ar_1(w) \cdots \ar_k(w) \r(w)$. Because $\al(w_1\cdots w_p)\subsetneq \Sigma$, either $w_1\cdots w_p$ is completely contained in an arch $\ar_i(w)$, or there is an arch $\ar_i(w)$ covering a prefix of $w_1\cdots w_p$ and the remaining suffix of $w_1\ldots w_p$ is either covered by $\ar_{i+1}$ or by $\r(w)$. In both cases, we can immediately see that replacing in $w$ the factor $w_1\cdots w_p$ by $u_2u_2 $  can only produce a word $w'$ with $\iota(w')\geq \iota(w)$ (we need to insert two times the factor $u_2$ to account for the letters of the prefix of $w_1\cdots w_p$ covered by $\ar_i(w)$ and for the letters of its suffix covered by by $\ar_{i+1}$ or by $\r(w)$, in the second case mentioned above). 
	
	For the converse, let $w$ be a word from $L(R_{i+1})$ such that $\iota(w)$ is maximum. Then $w=w' u_2 u_2 w''$, such that $w'$ is described by $U_1$ and $w''$ by $U_3$. For each letter $a\in \Sigma $ appearing in $U_2$, there exists a word $w_a$ which is described by $U_2$; if $u_2=a_1\cdots a_t$, let $w_{U_2}=w_{a_1}\cdots w_{a_t}$ and note that $w_{U_2} \in L(U_2^*)$. Now, $u=w' w_{U_2} w_{U_2}w''$ has $w$ as a subsequence, so $\iota(u)\geq \iota(w)$. Moreover, $u\in L(R_i)$. Our claim has now be proven.
	
	Moreover, if $R_0=U_1(U_2)^*\cdots U_{2p-1}(U_{2p})^*U_{2p+1}$, such that $U_{2i-1}$ is a valid star-free regex for $i\in [p+1]$, then this process is finished after $p$ steps, and the resulting regex $R_{p}$ has length upper bounded by $O(|R|)$. 
	
	The statement holds for $R'=R_p$.
	\qed\end{proof}

%% file: proofs/proof_FPT_sigma_regex.tex
\begin{proof}
	From $R$ we construct the star-free regex $R'$ as in Lemma \ref{lem:regex_acyclic}. This takes linear time $O(n)$, and $|R'|\in O(n)$. From the star-free regex $R'$ we can construct an acyclic NFA $A$ accepting $L(R')$ in $O(n)$ time (using the classical Thompson-construction, see, e.g., \cite{AllauzenM06}, and the references therein); $A$ has $O(n)$ states and $O(n)$ transitions. We can now apply Theorem \ref{thm:FPT_sigma} and get the result.
	\qed\end{proof}

%% file: proofs/proof_regex_is_hard.tex
\begin{proof}
The fact that $k$-ESU is in NP, irrespective of the value of $k$, follows from Lemmas \ref{lem:unbounded_universality_regex} and \ref{lem:regex_acyclic}. For a regular expression $R$, if there is no decomposition $R=U_1(U_2)^*U_3$ with $\Sigma\subseteq \al(U_2)$, then we construct a star-free regular expression $R'$, in linear time $O(|R|)$, such that $\max\{\iota(w)\mid w\in L(R')\}= \max\{\iota(w)\mid w\in L(R)\}$. Then we simply guess a word of length at most $|R'|$, check if this is generated by $R'$, and then check if its universality index is greater than $k$. 

For the hardness part, we present a reduction from $3$-SAT to $1$-ESU. Recall that the $3$-SAT problem takes a set of clauses $C_1, C_2, \dots, C_m$, each containing (at most) $3$ variables from the set $x_1, x_2, \dots, x_n$ of boolean variables. An instance of $3$-SAT is satisfiable if there exists some assignment to the set of variables such that every clause is true.

 Given an instance of $3$-SAT, we construct a regular expression $W$ as follows. We create the alphabet $\Sigma = \{c_1, c_2, \dots, c_m\}$ containing a unique symbol corresponding to each clause in the SAT instance. Given the variable $x_i$, let $C_{i,1}, C_{i, 2}, \dots, C_{i, j}$ be the clauses satisfied by a true assignment to $x_i$, and let $\overline{C}_{i,1}, \overline{C}_{i, 2}, \dots, \overline{C}_{i, j'}$ be the set of clauses satisfied by a false assignment to $x_i$. We construct, from $x_i$, the regular expression $V_i = (T_i \mid F_i)$, with  $T_i= c_{i, 1} c_{i, 2} \dots c_{i, j}$ and $F_i= \overline{c}_{i,1} \overline{c}_{i, 2} \dots \overline{c}_{i, j'}$, where $c_{i, \ell}$ is the symbol of $\Sigma$ corresponding to $C_{i, \ell}$, and $\overline{c}_{i, \ell'}$ the symbol of $\Sigma$ corresponding to $\overline{C}_{i, \ell'}$. We form the regular expression $W$ by concatenating $V_1, V_2, \dots, V_{n}$ into a single star-free regular expression, giving $W = V_1 V_2 \cdots V_n$. 
 
 Clearly, the length of the regular expression $W$ is $O(m)$, as each clause contains at most three variables; as side note, $W$ is star-free,  and there is an acyclic NFA of size $O(m)$ accepting $L(w)$. Moreover, the alphabet of $W$ has $m$ letters.

	Let now $W$ be a regular expression constructed for some given $3$-SAT instance as above. First, let us assume that there is some $1$-universal word $w'$ represented by $W$, further let $w' = v'_1 v'_2 \cdots v'_n$ where $v'_i$ is the factor of $w'$ corresponding to the clause $V_i$ of $w$. We create an assignment for the $3$-SAT instance by setting $x_i$ to true, if $v_i' = T_i$, or false if $v_i' = F_i$. Note that, as $w'$ is 1-universal, the symbol $c_i$ must appear in $v_j'$ for some $j \in [n]$, and by extension, this assignment must satisfy $c_i$. Therefore, the assignment constructed in this way will satisfy the $3$-SAT instance. 
	
	In the other direction, given an assignment to the variables $x_1, x_2, \dots, x_n$ satisfying the $3$-SAT instance, we construct a word $w' = u_1 u_2 \dots u_n$ where $u_i$ is either $T_i$, if $x_i$ is true, or $F_i$ if $x_i$ is false. As every clause in the $3$-SAT instance is satisfied, there must be some $i \in [n]$ such that $c_j \in \letters(u_i)$ for every $j \in [m]$. Therefore, $\iota(w') \geq 1$, completing the proof. \qed
\end{proof}

%% file: furtherdefs.tex

A $2$-$\exists$-universal NFA ${A}$, which is not $3$-$\exists$-universal, and a $1$-$\forall$-universal NFA ${B}$ which is not $2$-$\forall$-universal are shown in Figure~\ref{img:universal_automaton}.

\input{figure}

%% file: appendixproofs.tex


\noindent
\textbf{Lemma~\ref{lem:infUniv}.}
	$L$ contains, for each $i\in \N$, a word $w_i$ of universality $\iota(w_i)\geq i$ if and only if there exists a state $q$ of $A$ such that $\al(\labels(c_q))=\Sigma$. 

\input{proofs/proof_infUniv}
\noindent
\textbf{Lemma~\ref{lem:cycles}.}
We can compute in $O(m+n)$ time the sets $V_i\subseteq \Sigma\cup \{\epsilon\}$ of all labels of transitions of the strongly connected component ${\mathcal C}_i$ of $A$, for $i\in [e]$. We can then access in $O(1)$ the set $V_{S(q)}$ with $\al(\labels(c_q))=V_{S(q)}$, for all $q\in Q$.

\input{proofs/proof_cycles}
\noindent
\textbf{Lemma~\ref{lem:algo_one_seq}.}
	Given a sequence $q_1\cdots q_h$ of $h< n$ states with $q_h=f$ and $S(q_i)\neq S(q_j)$, for all $i\neq j$, we can compute in $O(n(m+ n\sigma \sqrt{n}))$ time the maximum integer $k$ for which there exists a walk $\beta$ with origin $q_1$ and target $q_h$ such that $\labels(\beta)$ is $k$-universal and $\beta$ goes through the strongly connected components of $q_1,\ldots,q_h$, and through no other strongly connected component.

\input{proofs/proof_algo_one_seq}
\noindent
\textbf{Theorem~\ref{thm:FPT_states}.}
	Given a regular language $L$, over $\Sigma$, with $|\Sigma|=\sigma$, specified as an NFA $A$ with $n$ states, we can solve $k$-ESU in $O(2^nn(m+ n\sigma\sqrt{n}))$ time. That is, in the case when the input is given as an NFA, $k$-ESU is FPT w.r.t. the number $n$ of states of the input NFA.

\input{proofs/proof_FPT_states}
\noindent
\textbf{Theorem~\ref{thm:FPT_sigma}.}
	Given a regular language $L$, over $\Sigma$, with $|\Sigma|=\sigma$, specified as an NFA $A$ with $n$ states, we can solve $k$-ESU in $O(n + m + m2^\sigma)$ time.

\input{proofs/proof_FPT_sigma}
\noindent
\textbf{Lemma~\ref{lem:regex_acyclic}.}
	For regular expression $R$, if there is no decomposition $R=U_1(U_2)^*U_3$ with $\Sigma\subseteq \al(U_2)$, then there exist $k\geq 1$ such that $L(R)$ is not $k$-universal. Moreover, there exists a star-free regular expression $R'$, which can be computed in linear time $O(|R|)$, such that $\max\{\iota(w)\mid w\in L(R')\}= \max\{\iota(w)\mid w\in L(R)\}$.

\input{proofs/proof_regex_acyclic}
\noindent
\textbf{Theorem~\ref{thm:FPT_sigma_regex}.}
	Given a regular language $L$, over $\Sigma$, with $|\Sigma|=\sigma$, specified as a regex $R$ of length $n$, we can solve $k$-ESU in $O(n2^\sigma)$ time.

\input{proofs/proof_FPT_sigma_regex}
\noindent
\textbf{Theorem~\ref{thm:regex_is_hard}.}
	$k$-ESU is NP-complete for languages defined by regular expressions. The problem is already NP-hard for $k=1$.

\input{proofs/proof_regex_is_hard}